\documentclass{amsart}

\usepackage{amsmath}
\usepackage{amsmath,amsfonts}
\usepackage[english]{babel}
\usepackage{graphicx}
\usepackage{url}

\newcommand{\cola}{\mathsf{a}}
\newcommand{\colb}{\mathsf{b}}
\newcommand{\colc}{\mathsf{c}}
\newcommand{\cold}{\mathsf{d}}

\newcommand{\col}{\chi}
\newcommand{\colt}{\tilde{\chi}}

\newcommand{\treecut}{\Psi}
\newcommand{\cut}{\kappa}
\newcommand{\bad}{\beta}
\newcommand{\totalbad}{\tau}
\newcommand{\distr}{\Delta}
\newcommand{\calA}{\mathcal{A}}
\newcommand{\powset}[1]{2^{#1}}
\newcommand{\NN}{\mathbb{N}}
\newcommand{\arrowset}{\mathcal{A}}
\newcommand{\play}{G}

\newcommand{\ourcomplex}{O(n \, \bad \, 2^{d+\bad} \, (d-1)^{d \bad / 2})}

\newcommand{\rk}{\operatorname{rk}}
\newcommand{\subrk}{\operatorname{subrk}}

\newcommand{\Bell}{\operatorname{Bell}}
\newcommand{\diam}{\operatorname{diam}}
\newcommand{\OPT}{\operatorname{OPT}}

\newcommand{\pplacer}{\textsf{pplacer}}

\newcommand{\rppr}{\textsf{rppr}}
\newcommand{\taxtastic}{\textsf{taxtastic}}

\newtheorem{lemma}{Lemma}
\newtheorem{cor}{Corollary}
\newtheorem{prop}{Proposition}
\newtheorem{theorem}{Theorem}
\newtheorem{problem}{Problem}
\newtheorem{defn}{Definition}

\newtheorem{alg}{Algorithm}

\newcommand{\eat}[1]{}

\newcommand{\forarxiv}[1]{#1}
\newcommand{\notforarxiv}[1]{}

\newcommand{\FIGTaxintro}{
\begin{figure}[ht]
\begin{center}
  \forarxiv{\includegraphics[width=6cm]{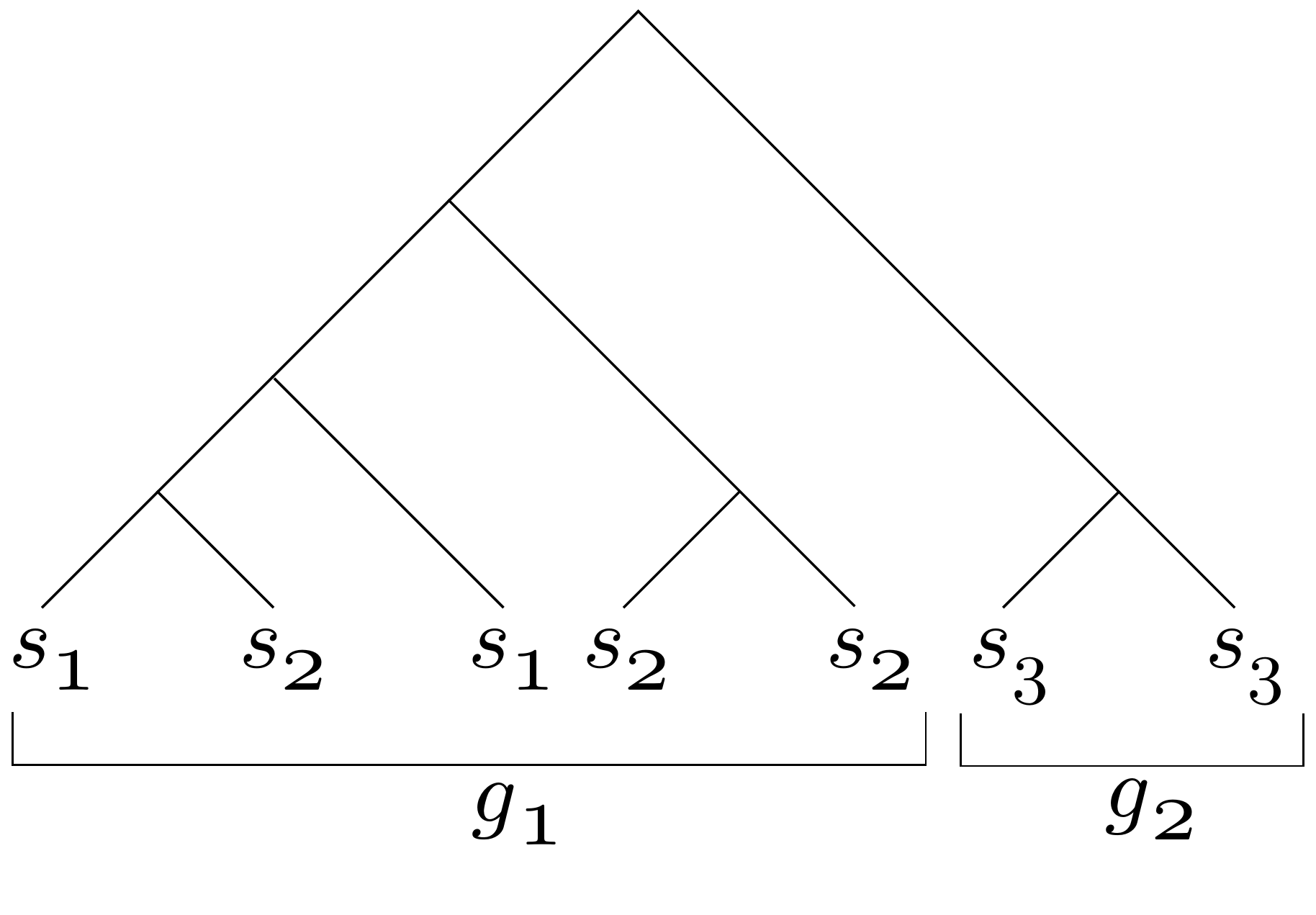}}
\end{center}
\caption{
A taxonomically labeled phylogenetic tree that is concordant with the genus level taxonomic assignments $g_i$ but not the species taxonomic assignments $s_i$.
}
\label{FIGTaxintro}
\end{figure}
}

\newcommand{\FIGAltConvex}{
\begin{figure}[ht]
\begin{center}
  \forarxiv{\includegraphics[height=3cm]{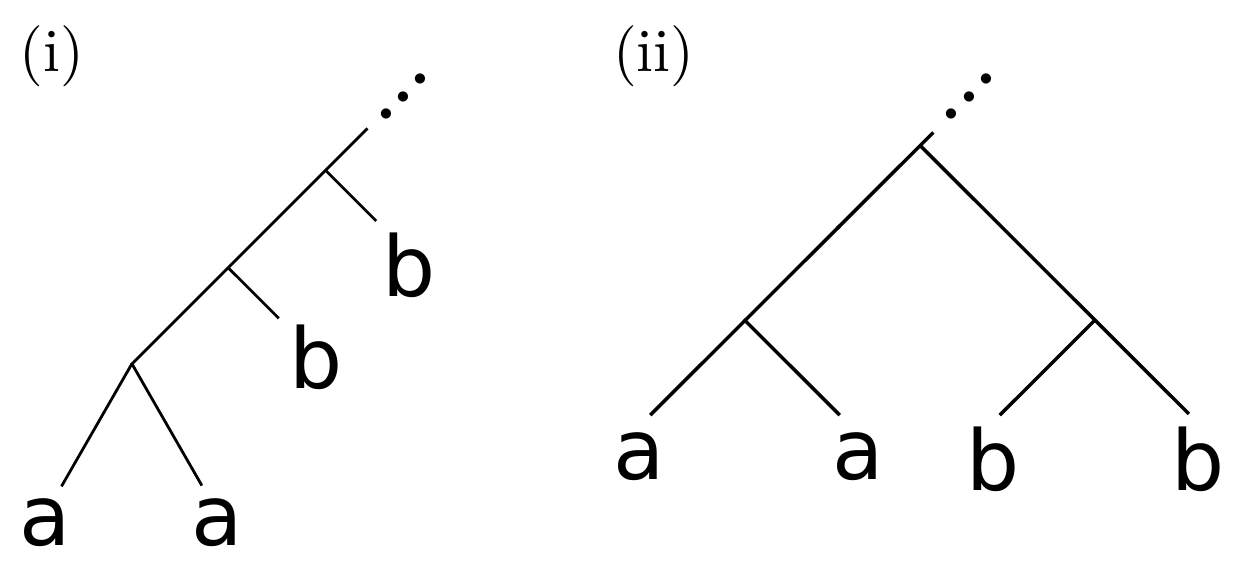}}
\end{center}
\caption{
Example colored trees showing the (i) original and (ii) strong definition of convexity, assuming $\cola$ and $\colb$ don't appear elsewhere in the tree.
In this figure, (i) and (ii) are convex according to the original definition, but only (ii) is convex according to the strong definition.
}
\label{FIGAltConvex}
\end{figure}
}

\newcommand{\FIGChallenge}{
\begin{figure}[ht]
\begin{center}
  \forarxiv{\includegraphics[height=5cm]{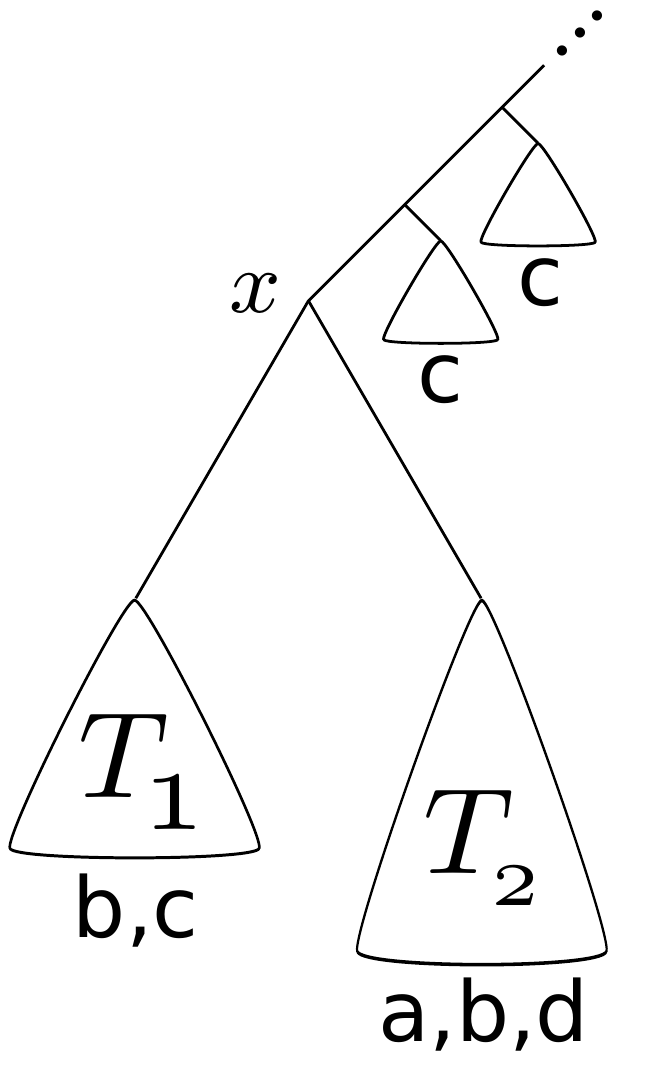}}
\end{center}
\caption{
A possible scenario encountered by the subcoloring recursion.
The letters $\cola,\colb,\colc$ and $\cold$ represent the presence of leaves with those taxonomic labels; asssume these taxonomic labels do not occur elsewhere in the tree.
The positions of colors $\colb$ and $\colc$ shows that this coloring is not convex, and a recursive subcoloring algorithm must decide at $x$ in which subtrees to allow the $\colb$ and $\colc$ colors.
}
\label{FIGChallenge}
\end{figure}
}

\newcommand{\FIGLocalBadness}{
\begin{figure}[ht]
\begin{center}
  \forarxiv{\includegraphics[width=10cm]{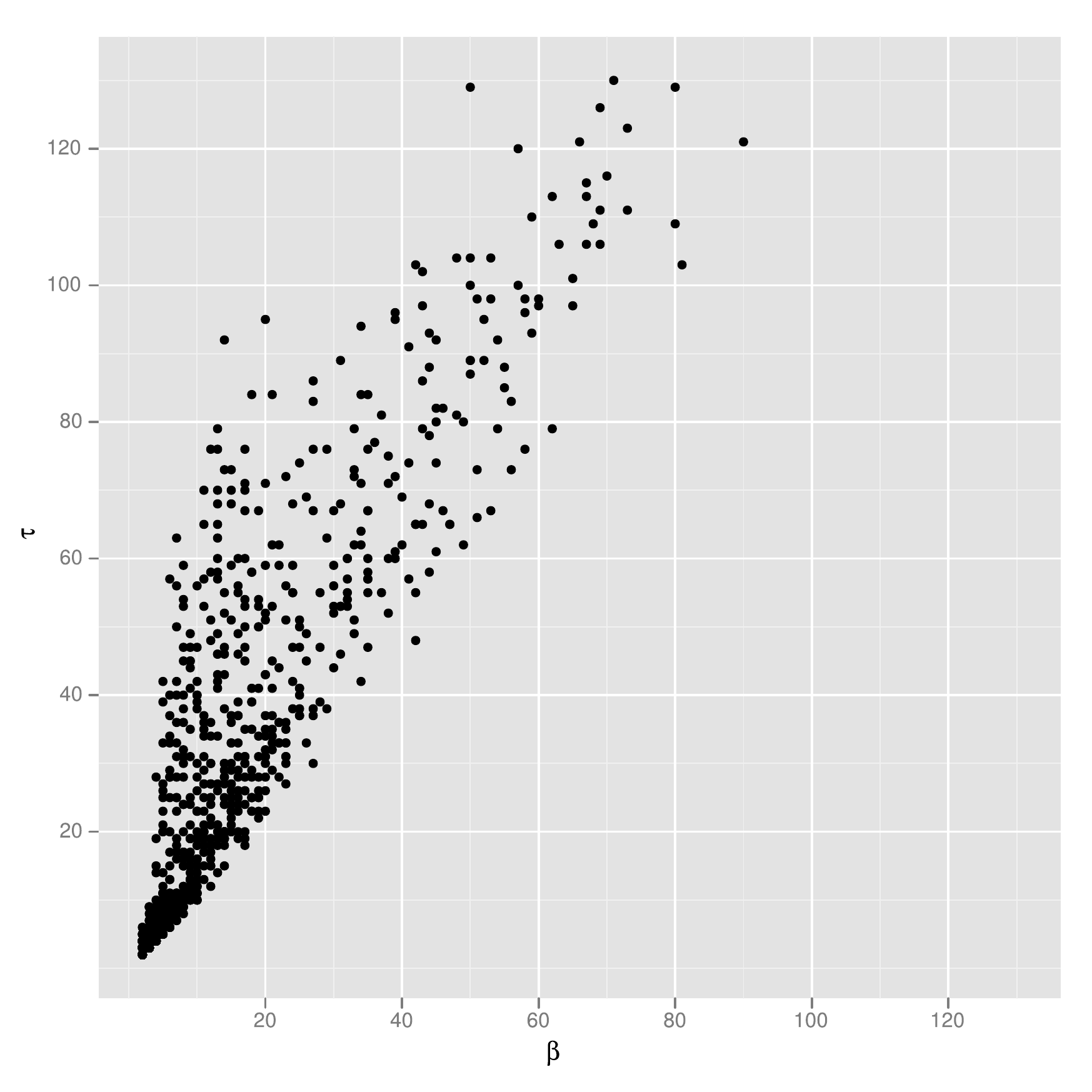}}
\end{center}
\caption{
The relationship between $\bad$, a local measure of nonconvexity, and $\totalbad$, a global measure, for our example data set.
Each point represents a single phylogenetic tree with taxonomic assignments at a given rank.
}
\label{FIGLocalBadness}
\end{figure}
}

\newcommand{\FIGImplicitConvex}{
\begin{figure}[ht]
\begin{center}
 \forarxiv{\includegraphics[width=6cm]{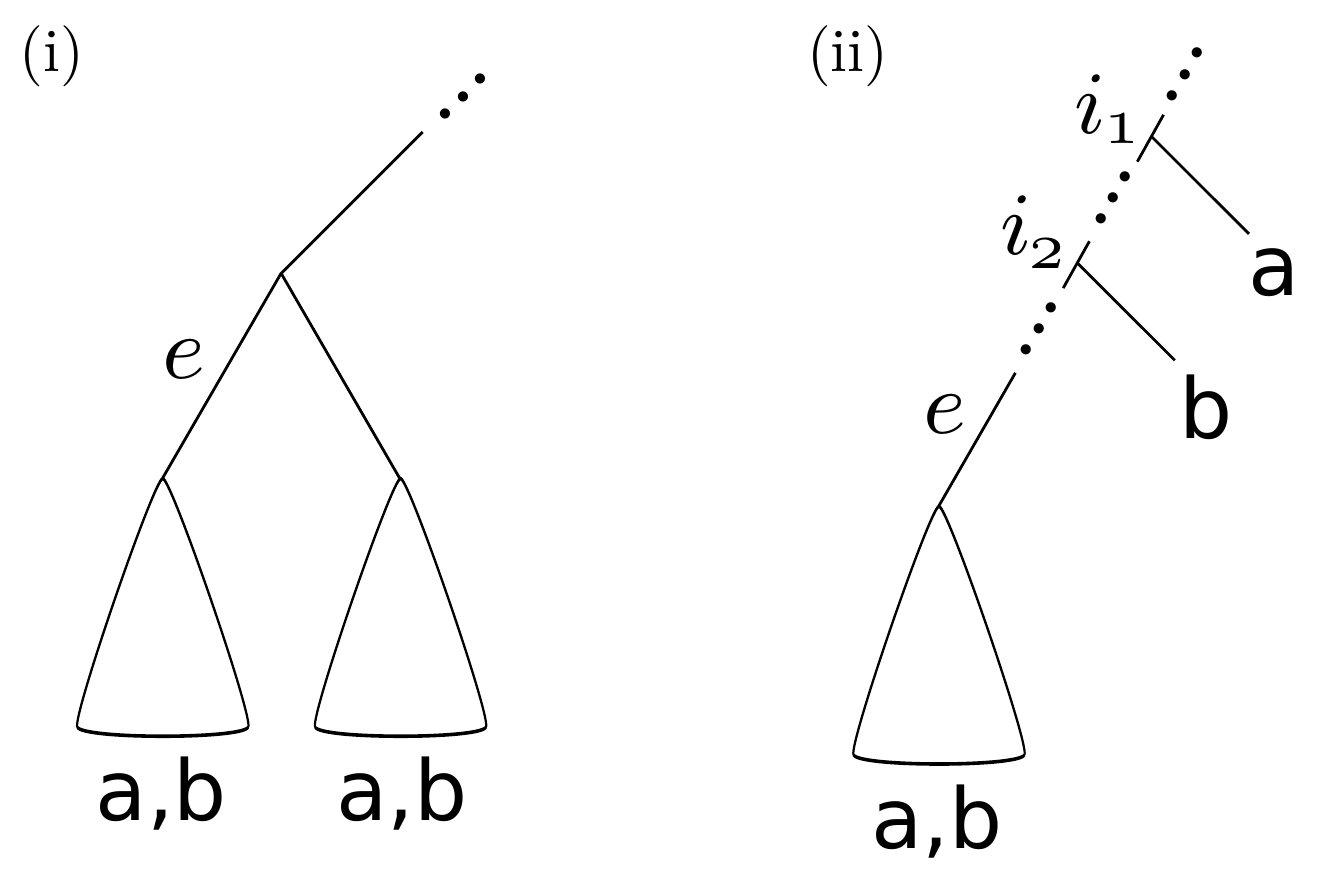}}
\end{center}
\caption{
Two potential settings for nonconvexity along an edge $e$ in the proof of Proposition~\ref{prop:implicitConvex}.
}
\label{FIGImplicitConvex}
\end{figure}
}

\newcommand{\FIGBBTiming}{
\begin{figure}[ht]
\begin{center}
 \forarxiv{\includegraphics[width=10cm]{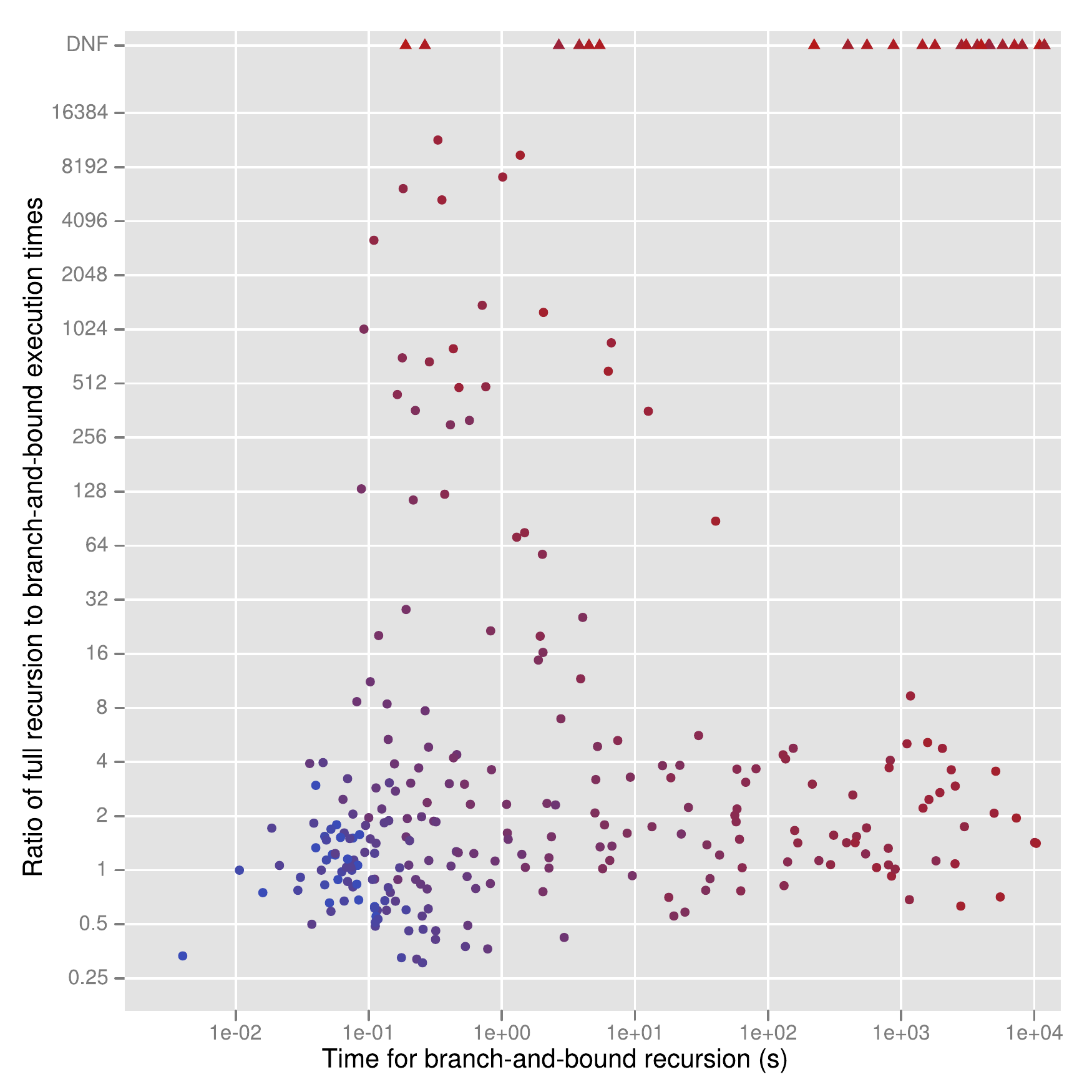}}
\end{center}
\caption{
Runtime comparison between the full recursion (Theorem~\ref{theorem:full}) to the branch and bound (Algorithm~\ref{alg:BB}).
``DNF'' means that the full recursion did not finish in the time and memory allotted.
Symbols colored according to their badness $\bad$, ranging from 4 (blue) to 14 (red).
}
\label{FIGBBTiming}
\end{figure}
}

\newcommand{\FIGBackarrow}{
\begin{figure}[ht]
\begin{center}
  \forarxiv{\includegraphics[width=6cm]{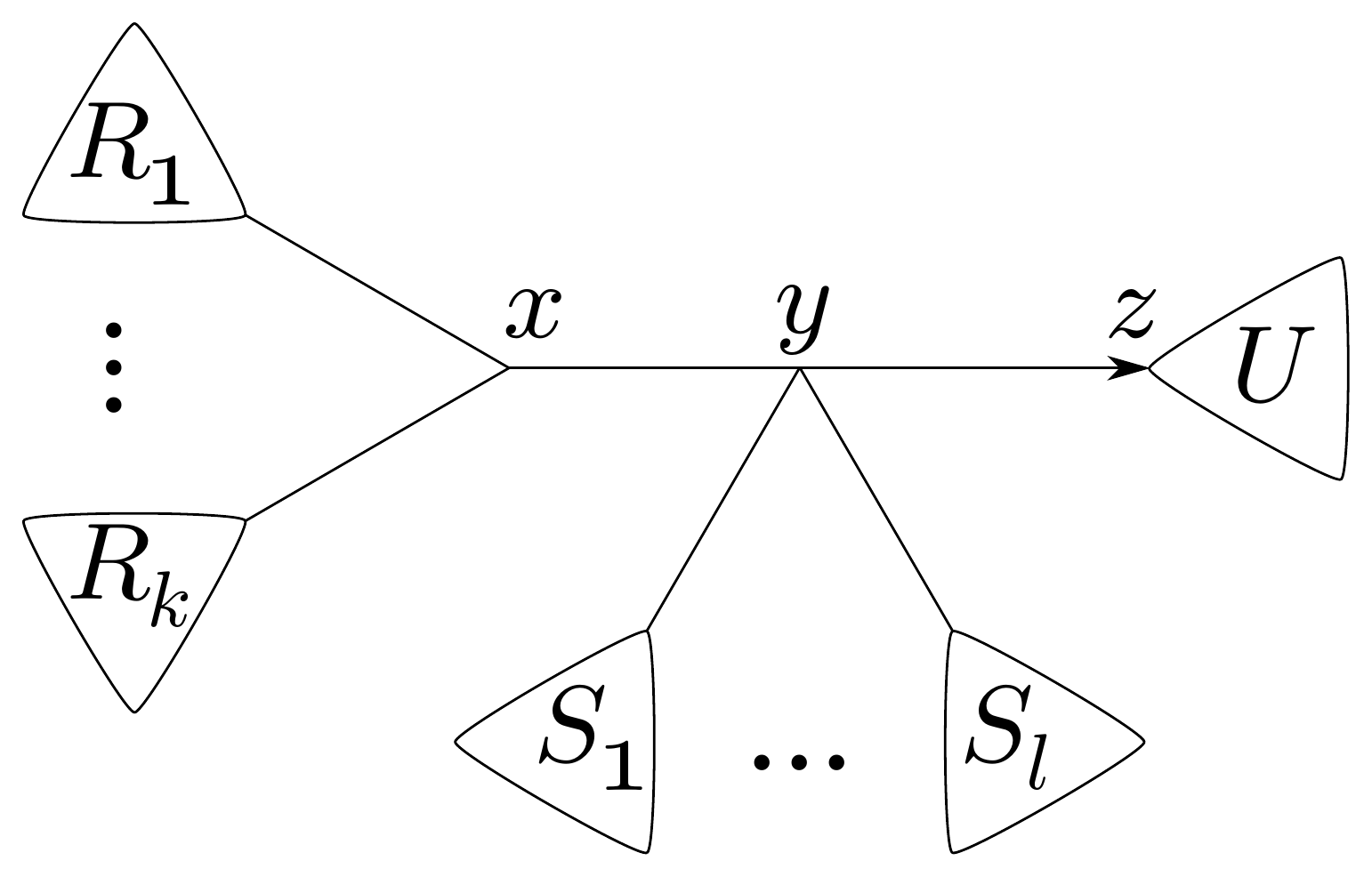}}
\end{center}
\caption{
Illustration of Lemma~\ref{lemma:backarrow}.
}
\label{FIGBackarrow}
\end{figure}
}
\begin{document}

\title[Reconciling taxonomy and phylogenetic inference]{Reconciling taxonomy and phylogenetic inference: formalism and algorithms for describing discord and inferring taxonomic roots}

\author{Frederick A Matsen}
\author{Aaron Gallagher}

\date{\today}

\begin{abstract}
Although taxonomy is often used informally to evaluate the results of phylogenetic inference and find the root of phylogenetic trees, algorithmic methods to do so are lacking.
In this paper we formalize these procedures and develop algorithms to solve the relevant problems.
In particular, we introduce a new algorithm that solves a ``subcoloring'' problem for expressing the difference between the taxonomy and phylogeny at a given rank.
This algorithm improves upon the current best algorithm in terms of asymptotic complexity for the parameter regime of interest; we also describe a branch-and-bound algorithm that saves orders of magnitude in computation on real data sets.
We also develop a formalism and an algorithm for rooting phylogenetic trees according to a taxonomy.
All of these algorithms are implemented in freely-available software.
\end{abstract}

\maketitle

\section{Introduction}

Since the beginnings of phylogenetics, researchers have used a combination of phylogenetic inference and taxonomic knowledge to understand evolutionary relationships.
Taxonomic classifications are often used to diagnose problems with phylogenetic inferences, and conversely, phylogeny is used to bring taxonomies up to date with recent phylogenetic inferences.
Similarly, biologists often evaluate a putative ``root'' of a phylogenetic tree by looking at the taxonomic classifications of the subtrees branching off that node.

This work is commonly done by hand.
That is, researchers who have knowledge of the taxa in their phylogenetic tree use their knowledge of the taxonomy to root the tree and describe the level of taxonomic concordance.
Despite the frequency with which these operations are done, we are not aware of any algorithms or software explicitly designed to address this problem.

In this paper we propose algorithms to express the discord between the taxonomy and the phylogeny and to root the tree taxonomically.
Our choice of algorithms will be guided by the parameter regime of relevance for modern molecular phylogenetics on marker genes: that of large bifurcating trees and a limited amount of discord with the taxonomy.

We state the agreement problem between the taxonomy and phylogeny in terms of an ``subcoloring'' problem previously described in the computer science literature \cite{MoranSnirConvexApprox07,MoranSnirConvexHard08}.
As described below, we make algorithmic improvements over previous work in the relevant parameter regime and present the first computer implementation.
For rooting, we show that the ``obvious'' definition has major defects when there is discordance between the phylogeny and the taxonomy at the highest multiply-occupied taxonomic rank.
We then present a more robust alternative definition and algorithms that can quickly find a taxonomically-defined root.

A related, but different, problem involves updating taxonomies based on phylogenetic inferences.
The most ambitious such project involves completely replacing the Linnean taxonomic system with a phylogenetic naming system, called PhyloCode \cite{forey2001phylocode}.
This proposal has met with substantial resistance from the community \cite{carpenter2003critique,nixon2003phylocode} and does not appear to have gained wide acceptance as of 2011.
Less ambitious but more highly accepted such projects include the Bergey \cite{kreig1984bergey} and Greengenes \cite{desantis2006greengenes} taxonomies; these have been curated to be more concordant with 16s phylogeny.
We do not approach the updating problem here, rather, we are interested in the commonly encountered problem of a researcher inferring a phylogenetic tree and wishing to understand the level of agreement of that tree with the taxonomy at various ranks and wishing to root the tree taxonomically.

\section{Expressing the differences between the taxonomy and the phylogeny}

\subsection{Informal introduction}

In this paper we will consider agreement with the taxonomy one taxonomic rank at a time, in order to separate out the different factors that can lead to discord between taxonomy and phylogeny.
These factors include phylogenetic methodology problems, horizontal gene transfer, lineage sorting, out of date taxonomic assignments, and mislabeling.
Various such factors lead to discordance at distinct ranks.
For example, we have observed rampant mislabeling at the species level in public databases, whereas higher-level assignments are typically more accurate.
Phylogenetic signal saturation or model mis-specification problems can lead to incorrect branching pattern near the root of the tree at the higher taxonomic levels, although the genus-level reconstructions can be correct.

An alternative to considering agreement one rank at a time would be to look for the largest set of taxa for which the induced taxonomy and phylogenetic tree agree on all levels.
Agreement between taxonomy and phylogeny at all taxonomic ranks simultaneously is equivalent to requiring complete agreement of the phylogeny and taxonomic tree.
Finding a subset of leaves on which two trees agree is known as Maximum Compatible Subtree (MCST), known to be polynomial for trees of bounded degree and NP-hard otherwise \cite{HeinEaComparingTrees96}.
Although such a solution is useful information, we have pursued a rank-by-rank approach here for the reasons described above.

\FIGTaxintro

We formalize the agreement of taxonomy with the phylogeny on a rank-by-rank basis in terms of \emph{convex colorings} \cite{MoranSnirConvexApprox07,MoranSnirConvexHard08}.
Informally, a convex coloring is an assignment of the leaves of a tree to elements of a set called ``colors'' such that the induced subtrees for each color are disjoint.
In this paper we will say that a phylogeny agrees with the taxonomic classification at rank $r$ if the taxonomic classifications at rank $r$ induce a convex coloring of the tree.
For example, in Figure~\ref{FIGTaxintro} the tree is not convex on the species level due to nonconvexity between species $s_1$ and $s_2$, although it is convex on the genus level, as the $g_1$ and $g_2$ genera fall into distinct clades.
In terms of the convex coloring definition, there is nontrivial overlap between the induced trees on $s_1$ and $s_2$.

We will express the level of discord between the taxonomy and the phylogeny at a rank $r$ in terms of the size of a \emph{maximal convex subcoloring}.
Given an arbitrary leaf coloring, a subcoloring is a coloring of some subset of the leaves $S$ of the tree that agrees with the original coloring on the set $S$.
The maximal convex subcoloring is the convex subcoloring of maximal cardinality.
For a tree that is taxonomically labeled at the tips, the discord at rank $r$ is defined to be the size of the maximal convex subcoloring when the leaves are colored according to the taxonomic classifications at rank $r$.

Our algorithmic contributions are twofold.
First, by developing an algorithm that only investigates removing colors when such a removal could make a difference, we show that the maximal convex subcoloring problem can be solved in a number of steps that scales in terms of a local measure of nonconvexity rather than the total number of nonconvex colors.
Second, we implement a branch and bound algorithm that terminates exploration early; this algorithm makes orders of magnitude improvement in run times for difficult problems.

\FIGAltConvex

Before proceeding on to outline how the algorithm works, note that the original definition of convexity is not the only way to formalize the agreement with a taxonomy at a given rank: for example, a stronger way of defining convexity is possible (Figure~\ref{FIGAltConvex}).
In this stronger version, colors must sit in disjoint rooted subtrees rather than in disjoint induced subtrees.
The algorithmic solution for this stronger version will be a special case of the previous one as described below.
It may be of more limited use for two reasons.
First, it depends on the position of the root: a tree that is strongly convex with one rooting may not be so in another.
Also, it is not uncommon for phylogenetic algorithms to return a tree like in Figure~\ref{FIGAltConvex}~(ii) although Figure~\ref{FIGAltConvex}~(i) may actually be correct tree; thus an algorithm that threw everything out except those that are convex in the strong sense might be overly strict.

\FIGChallenge

The purpose of this first half of the paper is to derive efficient algorithms for the convex subcoloring problem in the parameter regime of interest: a limited amount of discord for large trees.
The tree in Figure~\ref{FIGChallenge} serves to explain why the problem is combinatorially complex and motivates a recursive solution.
The idea of this solution is to recursively descend through subtrees, starting at the root.
Say this recursion has descended to an internal node $x$, and there are nodes of the color $\colc$ somewhere above $x$.
If the set of leaf colors in the subtree $T_1$ is $\{\colb,\colc\}$ and if the set of leaf colors in the subtree $T_2$ is $\{\cola,\colb,\cold\}$, then some removal of colors is needed due to nonconvexity between the $\colb$ and $\colc$ colors.
Assuming the leaf colors above $x$ are fixed, the choices are to uncolor the $\colc$ nodes in $T_1$ or the $\colb$ nodes in either $T_1$ or $T_2$.

One can think of ``allocating'' these colors to the subtrees: the possible choices are to allocate $\colc$ to $T_1$ but choose one of $T_1$ or $T_2$ to have $\colb$, or do disallow $\colc$ in $T_1$ and allow $\colb$ in both $T_1$ and $T_2$.
Here and in general, the crux of devising an efficient recursion is to efficiently decide which colors get allocated to which subtrees.
Convexity can be insured by explicitly choosing a color for each internal node, and making sure that the color allocations respect those choices in terms of convexity.

In fact, selecting these color allocations is the only problem, as a complete set of color allocations is trivially equivalent to a choice of coloring.
Indeed, given an optimal color allocation for each internal node, one can simply look at the allocations for the internal nodes just above leaves to decide whether those leaves get uncolored or not.
Conversely, given a leaf coloring, one can simply look at the color set of the descendants of that internal node to get the set of colors allocated to the subtrees.

In deciding the color allocations, we can restrict our serious attention to colors that are present on either side of an edge, such as $\colb$ and $\colc$ on either side of $T_1$'s root edge in Figure~\ref{FIGChallenge}.
We will say that these colors are \emph{cut} by the edge.
Colors that are not cut by an edge should not require any decision making when the recursive algorithm is visiting the node just above that edge.
However, doing the accounting is not completely straightforward: of the cut colors, one might only allocate $\colb$ to $T_2$, but $\cola$ and $\cold$ can be used as well.
Thus, allocations including some non cut colors must be considered, motivating the definition of the \emph{colors in play} (Definition~\ref{def:ColorsInPlay}).

Note that the ingredients of the decision being made in Figure~\ref{FIGChallenge} are the color specified by the coloring just above $x$ (in this case fixed to be $\colc$), and the colors available in the subtrees below $x$.
Given a set of colors to allocate to the leaves below $x$, the algorithm needs to decide how to allocate the possible colors to $T_1$ and $T_2$.
One way of doing that is to test every possible allocation using the previous results of the recursion and score them in terms of the size of the corresponding subcoloring.
Doing this with an awareness of the cut colors leads to an algorithm expressed in terms of the maximum number of colors cut by a given edge.

However, building such a comprehensive optimality map is not necessary.
By simply counting the number of leaves of each color below $x$, one can get upper bounds on the sizes of the corresponding subcolorings and only evaluate those that have the potential to be worth exploring.
This observation is the basis of the branch and bound algorithm (Algorithm~\ref{alg:BB}).

\subsection{Definitions and algorithms}

A \emph{rooted subtree} is a subtree that can be obtained from a rooted tree $T$ by removing an edge of $T$ and taking the component that does not contain the original root of $T$.
The \emph{proximal} direction in a rooted tree is towards the root, while the \emph{distal} direction is away from the root.
Given a tree $T$, let $N(T)$, $E(T)$, and $L(T)$ denote the nodes, edges, and leaves of $T$.
Given a set $U$, let $2^U$ denote the set of subsets of $U$.
If the tree is not rooted, root it arbitrarily.

Following the terminology of \cite{MoranSnirConvexApprox07, MoranSnirConvexHard08}, a \emph{color set} will be an arbitrary finite set.

\begin{defn}
  Let $T$ be a rooted tree, and let $F \subset L(T)$.
  A \emph{leaf coloring} is a map $\col: F \rightarrow C$.
\end{defn}

A color $c$ is \emph{cut} by an edge $e$ if there is at least one leaf of color $c$ on either side of $e$.
A \emph{multicoloring} is defined to be a map from the edges of the tree to subsets of the colors.

\begin{defn}
  Given a coloring $\col$ on a rooted tree $T$, an \emph{induced multicoloring} is the map $\colt: E(T) \rightarrow 2^C$
  such that $\colt(e)$ is the (possibly empty) set of colors cut by that edge.
\end{defn}

\begin{defn}
  Define the \emph{badness} $\bad(\col)$ of a coloring $\col$ to be $\max_{e \in E(T)} |\colt(e)|$.
  We say that a coloring is \emph{convex} if it has badness equal to zero or one.
\end{defn}

\begin{defn}
  The \emph{total number of bad colors} is $\totalbad(\col) = \left| \bigcup_{e \in \mathcal{E}} \colt(e) \right|$ where $\mathcal{E}$ is the set of edges where $|\colt(e)| \geq 2$.
\end{defn}

\begin{defn}
  A subcoloring of a leaf coloring $\chi: F \rightarrow C$ is a coloring $\chi': G \rightarrow C$ with $G \subset F$ such that $\chi'$ agrees with $\col$ on the domain of $\chi'$.
\end{defn}
These are partially ordered by inclusion of domains; the size of a subcoloring is defined to be the size of its domain.

\begin{problem}
\label{prob:subcolor}
  Given a leaf coloring $\col$ on a tree $T$, find a largest convex subcoloring.
\end{problem}

\subsubsection{Previous work and motivation for present algorithm}
The foundational work in this area was done by Moran and Snir \cite{MoranSnirConvexApprox07,MoranSnirConvexHard08}.
Their work is phrased in terms of ``convex recoloring,'' i.e.\ finding the minimal number of changes in a coloring in order to obtain one that is convex.

It suffices to consider subcolorings for the case of leaf colorings.
Indeed, any recoloring can be turned into a subcoloring by removing the color of all of the leaves that get recolored.
Conversely, any convex subcoloring can be turned into a convex recoloring in linear time \cite{bodlaender2007quadratic}.
For internal nodes, the color to be used for a given internal node is given by the definition of convex coloring.
For leaf nodes, simply take the color of the closest colored node.
In this equivalence, the number of leaves whose color is removed is equal to the number of leaves who get recolored; thus a minimal recoloring is equavalent to a maximal subcoloring.
We only consider subcolorings in this paper.

In \cite{MoranSnirConvexHard08}, Moran and Snir investigate both the general case of leaf colorings as well as the case of colorings including internal nodes.
They also consider non-uniform recoloring cost functions, where a ``cost'' is associated with recoloring individual nodes and the goal is to find a convex recoloring minimizing total cost.
In all settings, they demonstrate that the relevant recoloring problem is NP-hard.
They also demonstrate fixed parameter tractablity (FPT) of the problems, as described in the next paragraph.
In \cite{MoranSnirConvexApprox07} they present, among other results, a 3-approximation for tree recoloring.

The FPT bound for leaf coloring from \cite{MoranSnirConvexHard08}, $O(n^4 \, \totalbad \, \Bell(\totalbad))$, comes from an elegant argument using the Hungarian algorithm for maximum weight perfect matching on a bipartite graph.
In fact, an inspection of their proof reveals that their algorithm is $O(n \, d^3 \, \totalbad \, \Bell(\totalbad))$, where $d$ is the maximum degree of the tree.
$\Bell(k)$ denotes the $k$th Bell number, which is the number of unordered partitions of $k$ items; these numbers are known to satisfy the bounds $\left(\frac{k}{e \ln k}\right)^k < \Bell(k) < \left(\frac{k}{\ln k}\right)^k$ \cite{berend2010improved}.
Their recursion at a given internal node iterates over every unordered partition of the nonconvex colors, constructing a bipartite graph with edge weightings determined from the sizes of subcolorings of subtrees using those color sets for the set of excluded colors.
Applying the Hungarian algorithm to each such graph results in optimal solutions for every possible set of colors to exclude from the subtree at that internal node.
Because every unordered partition of the non-convex colors is considered, the algorithm is exponential in $\totalbad$.
For the case of general (i.e.\ not just leaf) colorings, Moran and Snir show that a dynamic program gives an $O(n \, \totalbad \, d^{\totalbad + 2})$ algorithm.
This of course also gives the same bound for leaf colorings.

The work of Moran and Snir was followed up by many authors.
For leaf-colored trees, Bachoore and Bodlaender \cite{bachoore2006convex} propose a collection of reductions to simplify the problem under investigation.
These reductions encode some of the logic of the algorithm presented here, such as that trees that have disjoint color sets can be solved independently.
They also use the fact that nonconvexity can be expressed in terms of the crossings of paths connecting leaves of the same color to show that the recoloring problem can be solved in $O(n 4^{\OPT})$ time, where $\OPT$ is the optimal number of uncolored leaves.
Note that this sort of bound is different than those described above, as $\OPT$ can get large even when the total number of bad colors is small.
The work for the general case culminated in the work of Ponta, H{\"u}ffner, and Niedermeier \cite{ponta2008speeding}, who use the childwise iterative approach to dynamic programing to construct an algorithm of complexity $O(n \, \totalbad \, 3^\totalbad)$.

\forarxiv{\FIGLocalBadness}

For trees built from real data, taxonomic identifiers are not randomly spread across the tree in a uniform fashion.
For example, species-level mislabeling will lead trees that are mostly convex with a couple of outliers, while a horizontal gene transfer will effectively ``transplant'' one clade into another.
In both of these situations there is a non-uniform distribution of taxonomic identifiers across the tree, and nonconvexity in these cases may be local.
Indeed, in Figure~\ref{FIGLocalBadness} we show the relationship between the badness $\bad$ and the total number of bad colors $\totalbad$ for our example trees, showing that the badness $\bad$ is significantly smaller than the total badness on a collection of phylogenetic trees for functional genes.
This motivates the search for a fixed parameter tractable algorithm that is exponential in $\bad$ rather than $\totalbad$.

Furthermore, phylogenetics is typically concerned with a setting of trees with small degree.
For example, many commonly used phylogenetic inference packages such as RAxML \cite{stamatakis2006raxml} and FastTree \cite{price2010fasttree} only return bifurcating trees; these sorts of programs are the intended source of trees for our convexify algorithm.
Even when multifurcations are allowed, the setting of interest is that of degree much smaller than $\bad$ or $\totalbad$, which has ramifications for algorithm choice as described below.

\subsubsection{Algorithm}

In this section we present our algorithm, which makes two improvements over previous work for the parameter regime of interest.
First, it only evaluates relevant colorings by restricting attention to cut colors, resulting in an algorithm that is exponential in $\bad$ rather than $\totalbad$.
% in contrast to \cite{MoranSnirConvexHard08}.
Still, such a complete recursion evaluates many sub-solutions that do not end up being used.
Because the problem is NP-hard, we cannot avoid some such evaluation, but we might hope to do better than evaluating everything.

This motivates the second aspect, a branch and bound strategy that can make orders of magnitude improvements in the run time of the algorithm.
In order to make the branch and bound algorithm possible, the algorithm enumerates all legal color allocations first, and ranks them according to the upper bound function.
By bounding the size of a solution for a given color allocation, we can avoid fully evaluating the sub-solution for that almost partition.
A simple way of bounding the size of a solution for an color allocation is the maximal size of the solution when convexity is ignored.

\begin{defn}
  Given a rooted subtree $T'$ of $T$, define $\cut(T')$ to be the colors of $\col$ cut by the root edge of $T'$ as it sits inside $T$.
\end{defn}

Assume that $T$ has been embedded in the plane, and that every internal node has been uniquely labeled.
For every such label $i$ let $t(i)$ be the ordered tuple of labels below $i$ in the tree, let $T(i)$ be the subtree below $i$, and use $\cut(i)$ as shorthand for $\cut(T(i))$.
Vector subscript notation will be used to index both $t(i)$ and the color set $k$-tuples defined next; i.e.\ $t(i)_j$ is the $j$th entry of $t(i)$.

\begin{defn}
 A \emph{color set $k$-tuple} is an ordered $k$-tuple of subsets of $C$.
 They will be denoted with $\pi$.
\end{defn}

These color set $k$-tuples will represent the allocation of colors to subtrees.
We will ensure convexity of these color allocations using the following two definitions.

\begin{defn}
Given a color set $k$-tuple $\pi$,
\[
A(\pi) = \bigcup_i \pi_i
\]
and
\[
B(\pi) = \bigcup_{i < j} (\pi_i \cap \pi_j)
\]
\end{defn}

\begin{defn}
  An \emph{almost partition} of $Y \subset C$ is an ordered 2-tuple $(b, \pi)$ where $b \in C$ and $\pi$ is a color set $k$-tuple such that $A(\pi) = Y$ and $B(\pi) \subset \{b\}$.
\end{defn}

These will be the color allocations at a given internal node $x$ with color $b$; this definition guarantees convexity locally.
As described in the Introduction, we would like to restrict our attention to cut colors, but this requires some attention as all of the colors used are not cut.
This motivates the following definition, which describes how the colors that are not explicitly excluded as the complement of $X$ are the ones available in the subtrees below $i$.

\begin{defn}
  \label{def:ColorsInPlay}
  Given $i$ an internal node index and $X \subset \cut(i)$ define
  \[
  \play(i, X) = X \cup \left( \bigcup_{j \in t(i)} \cut(j) \setminus (\cut(i) \setminus X) \right).
  \]
  $\play(i, X)$ will be called the \emph{colors in play} for $(i, X)$.
\end{defn}

\begin{defn}
  \label{def:ColorAllocation}
  Assume we are given an internal node $i$, $X \subset \cut(i)$, and $c \in C$.
  A \emph{legal color allocation} for $(i, c, X)$ is an almost partition $(b, \pi)$ of $\play(i,X)$ such that
  \begin{enumerate}
    \item $\pi_j \subset \cut(t(i)_j)$
    \item if $c \in X$ then $b=c$.
  \end{enumerate}
  Denote the set of such legal color allocations with $\distr(i,c,X)$, and let $\distr(i) = \bigcup_{c,X} \distr(i,c,X)$.
\end{defn}

These color allocations are exactly the set of choices that are allowed when developing a subsolution for a cut set $X$ such that the color $c$ is just above $X$.
The first condition ensures that the color choice sits inside the correct set of cut colors.
The second condition says that an internal node must take on any color found above and below it.

\begin{defn}
  \label{def:ImplicitSubcoloring}
  An implicit subcoloring for $T'$ is a choice of $(b(i), \pi(i)) \in \Delta(i)$ for every $i \in N(T')$ satisfying the following compatibility property for every $k \in t(i)$:
  \[
    (b(k), \pi(k)) \in \distr(k, b(i), \pi(i)_k).
  \]
\end{defn}
That is, the color allocation for every node descending from $i$ is a legal color allocation given the choices of $b(i)$ and $\pi(i)$ made at $i$.

As described in the Introduction, an implicit subcoloring defines an actual subcoloring via the implicit subcoloring just proximal to leaf nodes.
Indeed, say $t(i)_j$ is a leaf, and that $(b(i), \pi(i))$ is the color allocation for internal node $i$.
Then $\pi(i)_j$ is empty or a single element by definition; the color for leaf $t(i)_j$ is used in the subcoloring if $\pi(i)_j \neq \emptyset$.
Furthermore, every convex subcoloring can be written in this form.

\begin{prop}
  \label{prop:implicitConvex}
  Implicit subcolorings are convex.
\end{prop}

\begin{proof}
  Assume an implicit subcoloring $\{(b(j), \pi(j))\}_{j \in N(T)}$.
  Let $\col$ be a coloring coming from an implicit subcoloring.
  If $\chi$ is not convex, then there is an edge $e$ such that $|\colt(e)| \geq 2$.
  Say $\{\cola,\colb\} \subset \colt(e)$.
  Without loss of generality, the colors will be positioned as in one of the two cases depicted in Figure~\ref{FIGImplicitConvex}.
  In case (i), $|B(\pi(i))| \geq 2$, contradicting the definition of an almost partition.
  In case (ii), $b(i_1)$ is $\cola$ by the definition of almost partition because $\cola \in B(\pi(i_1))$.
  Then $b(i) = \cola$ for every $i$ between $i_1$ and $i_2$, inclusive, by part 2 of Definition~\ref{def:ColorAllocation} and Definition~\label{def:ImplicitSubcoloring}.
  However, $\colb \in B(\pi(i_2))$, contradicting the definition of almost partition.
\end{proof}

\FIGImplicitConvex

With this in mind, we can now speak of the size of an implicit subcoloring as the size of its associated convex subcoloring.
The goal, then, is to find the largest implicit subcoloring.

\begin{defn}
  Given internal node $i$ and $(b, \pi) \in \distr(i)$, a \emph{partial solution} for $(b, \pi)$ is an implicit subcoloring for $T(i)$ of maximal size such that the choice of almost partition for node $i$ is $(b, \pi)$.
\end{defn}

% \begin{prop}
%   Assume an internal node $i$ and $(b, \pi) \in \distr(i)$, and a \emph{partial solution} $\varphi := \{(b(j), \pi(j))\}_{j \in N(T(i))}$.
%   For any $k \in N(T(i))$, $\varphi$ restricted to $T(k)$ is a partial solution for node $k$ and $(b(k), \pi(k))$.
% \end{prop}
%
% \begin{proof}
%   If not, then there must exist an implicit subcoloring that is strictly larger than the restriction.
%   We could exchange $\varphi|_{T(k)}$ with this implicit subcoloring to get a partial solution that is strictly bigger than $\varphi$.
% \end{proof}
%
% This proposition implies a clear strategy: build up an implicit solution by finding partial solutions for subtrees.

\begin{theorem}
  \label{theorem:full}
  There is an $\ourcomplex$ complexity algorithm to solve the subcoloring problem for leaf labeled trees.
\end{theorem}

\begin{proof}
  For every internal node $i$, define the \emph{question domain} $Q(i)$  to be $C \times 2^{\cut(i)}$.
  An \emph{answer map} at internal node $i$ (resp.\ \emph{answer size map}) is a map $Y \rightarrow \distr(i)$ (resp.\ $Y \rightarrow \NN$) for some $Y \subset Q(i)$.

  We will fill out an answer map $\varphi_i$ and an answer size map $\omega_i$ as needed at every internal node $i$ as follows.
  For a given $i$, say we are given a question $(c,X) \in Q(i)$.
  If $i$ is a leaf, then $\varphi_i(c,X) = X$ and $\omega_i(c,X) = |X|$.
  Otherwise, say there are $\ell$ descendants of $i$.
  For each $(b,\pi) \in \distr(i,c,X)$, find the answers $\varphi_{t(i)_j} (b, \pi_j)$, and their associated $\omega_{t(i)_j}$ by recursion.
  Let
  \[
    \tilde{\omega}_i(b,\pi) = \sum_{1 \leq j \leq \ell} \omega_{t(i)_j} (b, \pi_j).
  \]
  Let $\omega_i(c,X)$ be the maximum value of $\tilde{\omega}_i(b,\pi)$ for $(b,\pi) \in \distr(i,c,X)$, and let $\varphi_i(c,X)$ be the $(b,\pi)$ obtaining this maximum.
  The result of this recursion after starting at the root with every color for $c$ will be a collection of answer maps for every $i$.

  These maps define an implicit subcoloring.
  This can be seen by descending through the tree recursively, using the $\varphi_i$ to get almost partitions from questions and passing the resulting questions onto subtrees.
  Specifically, for question $(c,X)$ at internal node $i$, let $(b(i), \pi(i)) := \varphi_i(c,X)$ then recur by passing question $(b(i)_j, \pi(i)_j)$ to $\varphi_{t(i)_j}$ for every descendant $j$.
  Start at the root, with index $\rho$, pick the color $c_\rho$ maximizing $\omega_\rho(c_\rho, \emptyset)$, and begin the recursion with $(c_\rho,\emptyset)$.

  This subcoloring is maximal.
  Assume any other convex subcoloring; this alternate subcoloring defines a question $(c,X)$ for each nonroot internal node $i$, where $c$ is the color of the edge above internal node $i$ induced by the definition of convexity and $X \subset \cut(i)$ is the set of cut colors that have leaves of that color below $i$.
  The collection of such questions also defines an answer $\varphi_i'(c,X)$ for every such $(c,X)$; by construction, the corresponding subsolution cannot be larger than that for $\varphi_i(c,X)$.

  The complexity estimate for a single internal node is as follows.
  The number of colors in play is bounded above by $d \bad / 2$, as each color in play must be cut in at least two edges.
  Thus, for a given question $(c,X)$, choosing the allocation can take $(d-1)^{d \bad / 2}$ steps for the colors other than $c$, while deciding where $c$ is present can take $2^d$ trials.
  For a given internal node there are at most $\bad \, 2^\bad$ choices of question, giving the bound.
\end{proof}

An upper bound for $\tilde{\omega}$ can be used to construct a branch and bound recursion as follows.

\begin{alg}[Branch and bound recursion to find optimal implicit subcoloring]
  Assume a function $\nu_i(b,\pi) \geq \tilde{\omega}_i(b,\pi)$ for all $(b, \pi) \in \distr(i)$.
  Proceed as in the proof of Theorem~\ref{theorem:full}, with the following modification.
  \label{alg:BB}
  For a given internal node $i$ with $c \in C$ and $X \subset \cut(i)$, find $\varphi_i(c, X)$ as follows:
  \begin{enumerate}
    \item sort the elements $(b, \pi)$ of $\distr(i, c, X)$ in decreasing size with respect to $\nu_i$.
    \item proceed down this ordered list as follows, starting with $q = 0$:
    \begin{enumerate}
      \item compute $\tilde{\omega}_i(b, \pi)$; if $q < \tilde{\omega}_i(b, \pi)$ then set $q = \tilde{\omega}_i(b, \pi)$
      \item call the next item in the ordered list $(b',\pi')$. If $q \geq \nu_i(b',\pi')$ then stop, otherwise recur to (a)
    \end{enumerate}
    \item let $\varphi_i(c, X)$ be the $(b, \pi)$ corresponding to $q$.
  \end{enumerate}

\end{alg}

A simple upper bound is the number of leaves that could be used given the restrictions in $\pi$ but ignoring convexity.
That is, let $\bar{\nu}_i(X)$ be the number of leaves of $T(i)$ with colors in $X$.
Then define $\nu_i(b,\pi) = \sum_j \bar{\nu}_{t(i)_j} (\pi_j)$.
This upper bound gives significant improvement in time used over the algorithm in Theorem~\ref{theorem:full} (Figure~\ref{FIGBBTiming}).

\forarxiv{\FIGBBTiming}

\subsection{Computer implementation}

The original algorithm described in Theorem~\ref{theorem:full} and the branch and bound algorithm in Algorithm~\ref{alg:BB} have been implemented in the \rppr\ binary of the \pplacer\ suite of programs \url{http://matsen.fhcrc.org/pplacer}.
The code is in written in OCaml \cite{ocaml}, an appropriate choice as it has $O(\log n)$ immutable set operations in the standard library.
The input can either be a ``reference package'' containing both taxonomic and phylogenetic information, or simply a phylogenetic tree along with a comma separated value file specifying the color assignments.
Our implementation has been validated using an independent ``brute-force'' implementation in Python; the two codes return identical results on a testing corpus consisting of all colorings on all trees of three to eight leaves with up to six colors.
The algorithm is called via a single command line call, which outputs a list of uncolored taxa for every nonconvex taxonomic rank as well as displaying them on a taxonomically labeled tree by highlighting them in red.

One time saving difference between the implementation and the algorithm from the previous section is that the computer implementation has a notion of ``no color.''
This is motivated by the fact that in the case when $c \not\in X$ and $B(\pi)$ is empty for an internal node $i$, there are a number of colorings of $i$ that will provide a convex subtree.
By collapsing all of the possible colors into a single ``no color'' in this case, we gain some savings in time and memory.

The ``no color'' version of the algorithm can also be used to solve the case of strong convexity described in the Introduction.
Specifically, restricting every internal node to have no color except for the internal nodes of subtrees that consist of entirely of one color leads to an algorithm for strong convexity.
This strong convexity version is available via a command line flag.

The data set used as a test set was a collection of 100 trees built from automatically recruiting sequences via a BLAST search via HMMs built from COG \cite{tatusov2003cog} alignments.
Taxonomic identifiers for the various ranks were found using the \taxtastic\ software, available at \url{http://github.com/fhcrc/taxtastic}.
Each trial was run three times and the results averaged; if any of the runs did not finish in 8 hours, exceeded 16G RAM usage, or encountered a stack overflow, the trial was marked as ``DNF.''
Every trial that completed according to these criteria using the full recursion also completed using the branch-and-bound.
Colored trees with badness strictly greater than 14 were excluded from Figure~\ref{FIGBBTiming}, as were trials that did not complete using either algorithm.
The full recursion and the branch and bound implementations only differ by a switch that controls if the algorithm terminates early.
Trials run on Intel Xeon (X5650) cluster nodes with 48G of RAM.
This test data set is available upon request.

\section{Taxonomic rooting}

Researchers generally like to root phylogenetic trees in a way that the progression along edges from the root to the leaves is one of descent.
There are a number of ways of achieving this, from using outgroups to using non-reversible models of mutation \cite{yap2005rooting}.
However, there has been surprisingly little work on one of the most commonly used informal means of rerooting, which is using taxonomic classifications.
By this, we mean looking for a location in the tree such that the leaf sets of the subtrees each have a single taxonomic classification at the highest taxonomic rank that contains multiple taxonomic identifiers.
Here we formalize this process and describe algorithms for finding the taxonomic root or roots.

\begin{defn}
\label{defn:rkUnion}
  A \emph{rank function} for a set $U$ is a map $\rk: \powset{U} \rightarrow \NN$ such that
  \[
  \rk(A \cup B) \geq \max(\rk(A),\rk(B))
  \]
  for all $A$ and $B$ in $\powset{U}$.
\end{defn}
It follows immediately that $\rk(A) \subset \rk(B)$ when $A \subset B \in \powset{U}$.
By an abuse of notation, we also let $\rk(T)$ signify $\rk(L(T))$ for a (sub)tree $T$ with leaf set in the domain of the rank function.
From a taxonomic perspective, $\rk(U)$ will represent the rank of the most specific taxonomic classification containing all of the taxonomic labels in $U$.
For example, the rank of a set consisting of a genus-level taxonomic assignment and an order-level one is the rank of order.
For this section, a \emph{taxonomically labeled phylogenetic tree} is one for which we have a rank function on the leaves.

Given $x$ a node of $T$, let $\treecut(x;T)$ represent the trees obtained by deleting $x$ from $T$.

\begin{defn}
  Define the \emph{subrank} $\subrk(x;T)$ to be $\max_{S \in \treecut(x;T)} \rk(S)$, the maximum rank of the subtrees of $T$ when rooted at $x$.
  We will say $x$ is a \emph{delicate taxonomic root} of $T$ if
  \[
    \subrk(x;T) = \min_{y \in N(T)} \subrk(y;T).
  \]
\end{defn}

This definition formalizes an intuitive definition of taxonomic root.
For example, imagine that we have a tree with the three domains of cellular life in three distinct subtrees: Bacteria, Archaea, and Eucaryota; call the internal node that sits between these subtrees $x$.
The subrank of $x$ is domain.
Any other internal node will contain some of each of the domains, and thus will have rank strictly higher than domain.
In this case, $x$ is the unique taxonomic root.

However, if the tree is not convex at the subrank of the delicate taxonomic root then every internal node will be a delicate taxonomic root; thus the ``delicate'' terminology.
Indeed, assume an internal node $y$ and $A, B \in \treecut(y;T)$ such that ${a_1, a_2} \subset L(A)$, ${b_1, b_2} \subset L(B)$, and $\subrk({a_1, a_2}) = \subrk({b_1, b_2}) = \subrk(T)$.
Then for any rooting, there must exist a subtree containing either $\{a_1, a_2\}$ or $\{b_1, b_2\}$, and the subrank must be equal to that of $T$.

We now develop a more robust definition of taxonomic root, which will require several definitions.
Think of edges of the tree as being unordered pairs $\{x,y\}$ of nodes.
\begin{defn}
  An \emph{arrow} on an edge $\{x,y\}$ is an ordered pair of nodes $(x,y)$.
  The first node of the pair is called the origin of the arrow, and the second is called the direction.
\end{defn}

\begin{defn}
  An \emph{arrow tree} $(T,\arrowset)$ is an ordered pair consisting of a tree $T$ and a set of arrows $\arrowset$ on the edges of $T$.
  A \emph{complete arrow tree} is an arrow tree such that for every node $x$ of the tree there is some arrow in $\arrowset$ with origin $x$.
\end{defn}

Note that $(x,y)$ and $(y,x)$ may both be part of an arrow set for a tree with an edge $\{x,y\}$.
We will use ``pointing towards'' and ``pointing away'' in their usual senses as they relate to arrows in the real world.

\begin{defn}
  The \emph{induced arrow tree} $(T,\arrowset)$ for a tree $T$ and a rank function $\rk$ is a complete arrow tree defined as follows.
  For a given internal node $x \in N(T)$, say $\{S_1,\cdots,S_n\} = \treecut(x;T)$ and let $r_i = \rk(S_i)$ for $1 \leq i \leq n$.
  Assume without loss of generality that $r_1 \leq r_2 \leq \cdots \leq r_n$.
  There is some minimal $1 \leq j \leq n$ such that $r_j = \cdots = r_n$.
  Let $\calA_x$ be
  \[
    \{ (x,y) \, | \, y \ \text{is the root of one of} \ S_j,\cdots,S_n \}.
  \]
  Then $\calA$ is the union of the $\calA_x$ for all nodes $x$.
\end{defn}
Intuitively, induced arrows point towards potential taxonomic roots.

\FIGBackarrow

\begin{lemma}
  \label{lemma:backarrow}
  Say $(T,\arrowset)$ is an induced arrow tree for a rank function $\rk$,
  and that $\{x,y\}$ and $\{y,z\}$ are adjacent edges of $T$.
  If $(y,z) \in \arrowset$ then $(x,y) \in \arrowset$.
\end{lemma}

\begin{proof}
  When $x$ is a leaf, $(x,y) \in \arrowset$ is automatic, thus assume it is not.
  Using terminology from Figure~\ref{FIGBackarrow}, because $(y,z) \in \arrowset$,
  \[
    \rk(R_1 \cup \cdots \cup R_k) \leq \rk(U).
  \]
  This implies that
  \[
    \rk(R_i) \leq \rk(U) \leq \rk(S_1 \cup \cdots \cup S_\ell \cup U)
  \]
  and thus that $(x,y) \in \arrowset$.
\end{proof}

Induction on the edges of a path shows the following:

\begin{cor}
  \label{cor:behind}
  \pushQED{\qed}
  Say $(T,\arrowset)$ is an induced arrow tree, and that
  $\{u,v\}$ and $\{x,y\}$ are edges of $T$ such that the path from $u$ to $y$ contains both $v$ and $x$.
  If $(x,y) \in \arrowset$ then $(u,v) \in \arrowset$.
  \popQED
\end{cor}

Informally, this corollary says that any time there is an arrow on edge $e_2$ pointing away from edge $e_1$, that there must be an arrow on $e_1$ pointing towards $e_2$.

\begin{defn}
  A multi arrow node (MAN) for a taxonomically labeled tree is a node $x$ such that there are two or more arrows in the induced arrow tree with $x$ as an origin.
\end{defn}

\begin{prop}
  \label{prop:man-point}
  Say $(T,\arrowset)$ is an induced arrow tree.
  If node $x$ is a MAN then for any node $y$ there must be an arrow in $\arrowset$ with origin $y$ pointing towards $x$.
\end{prop}

\begin{proof}
  Since $x$ is a MAN, there must be at least one arrow in $\arrowset$ with origin $x$ pointing away from $y$.
  This implies the proposition by Corollary~\ref{cor:behind}.
\end{proof}

Now imagine that $x$ and $z$ are two MANs, and $y$ is on the path between $x$ and $z$.
By the above proposition, there must be arrows with origin $y$ pointing towards both $x$ and $z$, showing that $y$ will be a MAN.
Thus:

\begin{prop}
  \pushQED{\qed}
  MANs form a connected set.
  \popQED
\end{prop}

\begin{defn}
  An edge $\{x,y\}$ is a \emph{bi-arrow edge} of an arrow set $\arrowset$ if $(x,y)$ and $(y,x)$ are in $\arrowset$.
\end{defn}

\begin{prop}
  If the set of MANs is empty, then there is exactly one bi-arrow edge.
\end{prop}

\begin{proof}
  First note that there cannot be two or more bi-arrow edges when the set of MANs is empty; in that case by Corollary~\ref{cor:behind} there would have to be a MAN between them.
  Now assume there are no bi-arrow edges.
  Since the set of MANs is empty, then for every leaf of the tree the sequence of nodes determined by following arrows is well defined.
  Note that the arrow on every leaf is pointing into the interior of the tree, and thus the sequence of nodes starting from an arbitrary leaf cannot hit another leaf.
  Therefore the sequence of nodes must backtrack somewhere, contradicting that there are no bi-arrow edges.
\end{proof}

\begin{defn}
  Assume a taxonomically labeled tree $T$.
  If there is at least one MAN then define the set of taxonomic roots to be the set of MANs.
  Otherwise define it to be the set of nodes of the bi-arrow edge.
\end{defn}

Let $\diam(T)$ be the node-diameter of $T$, i.e.\ the number of steps from edge to edge required to traverse the tree.
Because every arrow with a non-root origin points in the direction of the taxonomic roots:

\begin{prop}
  \pushQED{\qed}
  A taxonomic root for a tree $T$ with $n$ leaves can be found in at most $\diam(T)$ steps.
  \popQED
\end{prop}

\subsection{Computer implementation}
Taxonomic rerooting has been implemented in the \rppr\ binary of the \pplacer\ suite of programs \url{http://matsen.fhcrc.org/pplacer}.
However, rather than finding all possible taxonomic roots as described above, the program reports one of the roots after applying the maximal subcoloring algorithm as described in the previous section to the highest multiply occupied taxonomic rank.
Such a root is the closest approximation to the one ``best'' taxonomic root in the presence of nonconvexity.

\section{Conclusions and future work}

We have formalized the question of describing the discordance of a phylogenetic tree with its taxonomic classifications in terms of a convex subcoloring problem previously described in the literature.
This coloring problem has some elegant solutions for the general case, but the parameter regime of interest here consists of trees of small degree and local nonconvexity.
These considerations motivate a solution that solves a given recursion for as few ``questions'' as possible.
The first component of this is to restrict attention to cut colors, resulting in a smaller base for the exponential complexity (Figure~\ref{FIGLocalBadness})
The second is a branch and bound algorithm that gives a significant improvement in runtime compared to the algorithm in Theorem~\ref{theorem:full} (Figure~\ref{FIGBBTiming}).
To enable this the $\varphi_i$ are only built up ``upon demand,'' that is, when a given question is asked.
The implementation described here is the first of which we are aware, and certainly the first that conveniently integrates with taxonomic annotation.

We also develop the first formalism for taxonomic rooting of phylogenetic trees, show that the original definition is useless in the presence of nonconvexity, and develop a more useful definition.
This version can be found in time linear in the diameter in the tree.

We are currently developing a computational pipeline to reclassify sequences in public databases based on these algorithms.

We are also using these algorithms together to develop a collection of automatically curated ``reference packages'' that bring together taxonomic and phylogenetic for the purposes of environmental short read classification, visualization, and comparison.

\section{Acknowledgements}

This work was motivated by joint work with David Fredricks, Noah Hoffman, Martin Morgan, and Sujatha Srinivasan at the Fred Hutchinson Cancer Research Center.
We are especially grateful to Noah Hoffman for providing feedback on early results of the algorithm, to Shlomo Moran and Sagi Snir for help understanding their algorithm, and to Robin Kodner for allowing the COG trees to be used as test data for our algorithm.
Both authors were supported by NIH grant HG005966-01.

\bibliography{taxiphy}
\bibliographystyle{plain}

\end{document}